\documentclass[12pt]{article}
\pdfoutput=1 
\usepackage[T1]{fontenc}
\usepackage[utf8]{inputenc}
\usepackage[english]{babel}
\usepackage{hyperref}
\usepackage{siunitx}
\usepackage{multicol}
\usepackage{graphicx}
\usepackage{amsmath}
\usepackage{amssymb}
\usepackage{multirow}
\usepackage{bbm}
\usepackage{bm}
\usepackage{authblk}
\usepackage{color}
\usepackage[round]{natbib}
\usepackage{fancyhdr} 
\usepackage{subfig}
\usepackage{float}
\usepackage{textgreek}
\usepackage{mathtools}
\usepackage{paralist}
\usepackage{adjustbox}
\usepackage{bbm}
\usepackage{booktabs,ctable,threeparttable}
\usepackage{enumitem}
\usepackage{amsthm}

\newcommand{\diagdots}[3][-25]{%
	\rotatebox{#1}{\makebox[0pt]{\makebox[#2]{\xleaders\hbox{$\cdot$\hskip#3}\hfill\kern0pt}}}}

\newcommand\norm[1]{\left\lVert#1\right\rVert}

\theoremstyle{plain} 
\theoremstyle{plain} 
\theoremstyle{plain} 
\theoremstyle{plain} \newtheorem{lemma}{Lemma}
\theoremstyle{plain} \newtheorem{corollary}{Corollary}
\theoremstyle{remark}

\usepackage{paralist} 
\DeclareMathAlphabet{\mathpzc}{OT1}{pzc}{m}{it}

\allowdisplaybreaks

\date{\vspace{-5ex}}
\begin{document}
	\def\spacingset#1{\renewcommand{\baselinestretch}%
		{#1}\small\normalsize} \spacingset{1}
	\title{\bf Simple bootstrap for linear mixed effects under model misspecification}
	\author{Katarzyna Reluga\thanks{Division of Biostatistics, School of Public Health, 		University of California, Berkeley, U.S.A.  E-mail: katarzyna.reluga@berkeley.edu.}\;
	and\;Stefan Sperlich\thanks{Stefan Sperlich is a Professor at the Geneva School of Economics and Management, University of Geneva, Switzerland. E-mail: stefan.sperlich@unige.ch.\\
   {{The authors gratefully acknowledge support from} the Swiss National Science Foundation for the projects  200021-192345 and P2GEP2-195898. \\We would like to thank D. Flores, W. Gonzalez Manteiga, E. L\'opez-Vizca\'{\i}no, D. Morales, T. Schmid, N. Salvati and S. Ranjbar for helpful discussions.}}
	}
	\maketitle

	\begin{abstract}
		\noindent
		Linear mixed effects are considered excellent predictors of cluster-level parameters in various domains. However, previous work has shown that their performance can be seriously affected by departures from modelling assumptions. Since the latter are common in applied studies, there is a need for inferential methods which are to certain extent robust to misspecfications, but at the same time simple enough to be appealing for practitioners. We construct statistical tools for cluster-wise and simultaneous inference for mixed effects under model misspecification using straightforward semiparametric random effect bootstrap. In our theoretical analysis, we show that our methods are asymptotically consistent under general regularity conditions. In simulations our intervals were robust to severe departures from model assumptions and performed better than their competitors in terms of empirical coverage probability.  
	\end{abstract}
	
	\noindent%
	\vfill
	
	\newpage

\section{Introduction}\label{sec:intro}
Linear mixed models are frequently used for modelling hierarchical and longitudinal data. Within this modelling framework, population parameters are represented using fixed regression parameters, whereas the extra between-cluster variation is captured by cluster-specific random effects. We consider bootstrap methods for statistically valid inference for mixed effects which are linear combinations of fixed and random effects. Mixed effects are considered excellent predictors of cluster-level parameters in various domains, e.g. small area estimation, ecology or medicine \citep[cf. monographs of][]{verbekeMolenberghs2000,Jiang2007,rao2015small}. 

Further inference on mixed parameters heavily depends on model and distributional assumptions. Bootstrap methods have been introduced to partially relax this reliance and approximate in a flexible way functions of the estimators and predictors. Although they could be derived analytically using model-dependent large sample theory, the application of the latter often leads to inaccurate results in finite samples, and it is typically not robust to model misspecifications \citep[cf.,][]{chatterjee2008parametric, reluga2021simultaneous}. 

The family of bootstrap methods for clustered data is rich, and the extensive reviews are provided by \cite{field2007bootstrapping,chambers2013random} and more recently \cite{flores2019bootstrap}. All essential procedures can be classified into three broad categories: bootstrapping by resampling clusters and observations within clusters \citep{davison1997bootstrap,mccullagh2000resampling}, bootstrapping by random weighting of estimating equations \citep{field2010bootstrapping,samanta2013bootstrapping,o2018bootstrapping}, and bootstrapping by resampling predictors of random effects and/or residuals \citep{davison1997bootstrap}. The latter is referred to as a random effect bootstrap and can be further subcategorized into parametric versions \citep{butar2003measures,hall2006parametric,chatterjee2008parametric} and semiparametric versions \citep{carpenter2003novel,hall2006nonparametric,lombardiasperlich2008,opsomer2008splines}. Regardless of the category they belong to, the main goal of all bootstrap schemes is to construct the empirical estimates which faithfully reproduce some features of the true data generating mechanism. There exist a range of criteria to evaluate the quality of bootstrap schemes for clustered data. In the context of inference for mixed parameters, the existing literature focuses on bootstrap estimation of the mean squared error which boils down to the accurate approximation of the first few moments \citep[see, e.g.][]{butar2003measures,hall2006nonparametric,chatterjee2008parametric}. In our work, we assess the ability of bootstrap methods to reproduce cumulative distribution functions of some continuous functions of mixed effects which are used in the subsequent steps of statistical inference. At this place we need to emphasize that our goal is not to compare the performance of all existing procedures to select an optimal scheme with respect to a predefined criterion. Even though such a comparison in the context of mixed effects has not been attempted yet and it could be an interesting direction for further research, it requires a careful definition of the optimality criterion which is beyond the scope of this manuscript

In this article, we construct statistical tools for cluster-wise and simultaneous inference for mixed parameters under model misspecification using simple, semiparametric random effect bootstrap as in \cite{carpenter2003novel} and \cite{opsomer2008splines}. We show that our bootstrap scheme successfully reproduces cumulative distribution functions of studentized and maximal statistics which are the core elements of our inferential tools. We thus generalize the work of \cite{reluga2021simultaneous} who develop inferential tools for linear mixed effect once the modelling assumptions are satisfied. Our theory applies to the construction of intervals and testing procedure. In our analysis, we show that our methods are asymptotically consistent under general regularity conditions. In simulations our intervals were robust to severe departures from model assumptions and performed better than their competitors in terms of empirical coverage probability. Our bootstrap-based inference is complementary to other techniques handling model misspecifications and dealing with outliers, such as robust inference \citep{ChambersEtal2006,sinha2009robust} or estimation using data transformation \citep{RojasEtal2020}. 
	
\section{Inference on linear mixed effects}\label{sec:lmm}

Consider a response vector  $y\in \mathbb{R}$ modelled by $y=X\beta+Z u+e$ where  $X\in \mathbb{R}^{n\times(p+1)}$, $Z\in \mathbb{R}^{n_j\times q}$ are known full column rank design matrices for fixed and random effects, vector $\beta\in \mathbb{R}^{p+1}$ contains fixed effects, whereas random effects $u\in \mathbb{R}^{q}$ and errors $e\in \mathbb{R}^{n}$ are assumed to be mutually independent and identically distributed with $var(e) = G$ and $var(u)=R$. We focus on the model of \cite{laird1982random}
\begin{equation}\label{eq:LMM_b}
	y_j=X_j\beta+Z_ju_j+e_j,\quad j=1,\dots,m,
\end{equation}
where $y_j\in \mathbb{R}^{n_j}$, $X_j \in \mathbb{R}^{n_j\times(p+1)}$, $Z_j\in\mathbb{R}^{n_j\times q_j}$, $e = (e_1, e_2, \dots, e_m)^T$,  $u = (u_1, u_2, \dots, u_m)^T$. We denote the total sample size with $n$, the number of clusters with $m$ and $n=\sum_{j=1}^{m}n_j$ where  $n_j$ is the number of observations in the $j^{th}$ cluster. 
Furthermore, $G$ and $R$ are block-diagonal with blocks $G_j=G_j(\delta)\in\mathbb{R}^{q_j\times q_j}$ and $R_j=R_j(\delta)\in\mathbb{R}^{n_j\times n_j}$ which depend on variance parameters $\delta=(\delta_1,...,\delta_h)^T$. Let $E(y) = X\beta$ and $var(y) = V =  R + Z GZ^T$ where $V$ is a block-diagonal with blocks $V_j = R_j + Z_j GZ_j^T$. Under normality of random effects and errors, $y_j\sim N(X_j\beta, V_j)$ and $y_j|u_j\sim N(X_j\beta + Z_ju, G_j)$. The methods of maximum likelihood and restricted maximum likelihood are often used to obtain an estimator $\hat{\delta} =(\hat{\delta}_1,\dots,\hat{\delta}_h)^t$ \citep[see, for example,][Chapter 5]{verbekeMolenberghs2000}. In contrast, $\beta$ and $u$ are estimated and predicted using two-stage techniques. In particular, in the first stage one can use maximum likelihood, estimating equations of \cite{henderson1950} or h-likelihood of \cite{lee1996hierarchical} to obtain best unbiased linear estimator $\tilde{\beta} =\beta(\delta)= (X^tV^{-1}X)^{-1}X^tV^{-1}y $ and the best unbiased linear predictor $\tilde{u}_j=u_j(\delta)=G_j Z^t_jV_j^{-1}(y_j-X_j \tilde{\beta})$. In the second stage, we replace $\delta$ with $\hat{\delta}$ which results in empirical best unbiased linear estimator $\hat{\beta}=\beta(\hat{\delta})$, and empirical best unbiased linear predictor $\hat{u}_j=u_j(\hat{\delta})$. Our goal is to develop valid inferential tools for general cluster-level parameters 
\begin{equation}\label{eq:theta}
	\theta_j = k^T_j\beta+l^T_ju_j,\quad j=1,\dots,m,
\end{equation}
with known $k_{d}\in \mathbb{R}^{p+1}$ and $l_j\in \mathbb{R}^{q_j}$. The application of the two-stage approach leads to 
\begin{equation*}
	\tilde{\theta}_j = \theta_j(\delta) = k^T_j\tilde{\beta}+l^T_j\tilde{u}_j,\quad \hat{\theta}_j = \theta_m(\hat{\delta}) = k^T_j\hat{\beta}+l^T_j\hat{u}_j \quad  j = 1, \dots,m.
\end{equation*}

We focus on the development of methods to construct $1-\alpha$ confidence (or prediction) intervals and carry out hypothesis testing for mixed parameter $\theta_j$ following the ideas of \cite{reluga2021simultaneous}. Let $\sigma^2_j = var(\hat{\theta}_j)$ be a general estimator of variability of mixed effect and $\hat{\sigma}^2_j$ its estimated version. We define a t-statistic and a maximal statistic as follows:
\begin{equation}\label{eq:statistics}
	t_j = \frac{\hat{\theta}_j- \theta_j}{\hat{\sigma}_j}, \quad M = \max_{j=1,\dots m}\left\lvert  t_j\right\rvert, \quad j = 1, \dots, m.
\end{equation}
Individual confidence interval $I_{j, \alpha}$ at $1-\alpha$-level for $\theta_j$ in \eqref{eq:theta} is a region which satisfies $P(\theta_j \in I_{j, \alpha})=1-\alpha$. To construct $I_{j, \alpha}$, it is enough to find a critical value which is a high quantile from the distributions of statistic $t_j$, that is $q_{j,\alpha}=\inf\{a\in \mathbb{R}: P(t_j\leq a)\geq 1 -\alpha\}$. We can use a similar strategy to construct simultaneous confidence intervals $I_{\alpha}$ at $1-\alpha$-level which satisfies $P(\theta_j \in I_{\alpha}\; \forall j\in[m])=1-\alpha$, where $[m] = \{1, \dots, m\}$. Let $q_{\alpha}=\inf\{a\in \mathbb{R}: P(M \leq a)\geq 1 -\alpha\}$ be a high quantile from the distribution of statistic $M$. We thus have
\begin{equation}\label{eq:ind_sim_interval}
	I_{j, \alpha}: \big\{\hat{\theta}_j  \pm q_{j,\alpha}\times \hat{\sigma}_j\big\}, \quad 	I_{\alpha} = \bigtimes_{j=1}^m 	I^s_{j,\alpha}, \quad 
	I^s_{j,\alpha} : \big\{\hat{\theta}_j  \pm q_{\alpha}\times \hat{\sigma}_j\big\},
\end{equation}
and it follows that $I_{\alpha}$ covers all mixed effects with probability $1-\alpha$ \cite[see][for more details on the importance of maximal statistic in the simultaneous inference for mixed parameters]{reluga2021jasa,reluga2021simultaneous}. Due to the central limit theorem, $q_{j,\alpha}$ is often replaced by a high quantile from the standard normal distribution or the Student's t-distribution. The relation between confidence intervals and hypothesis testing allows us to define modified statistics $t_j$ and $ M$ that can be used to carry out hypothesis testing. More specifically, let $A\in \mathbb{R}^{m'\times m}$,  $\theta_{H} = (\theta_{H_1}, \theta_{H_2}, \dots, \theta_{H_{m'}}) = A\theta \in \mathbb{R}^{m'}$ and $c = (c_1, c_2, \dots, c_{m'})=\mathbb{R}^{m'}$ be a vector of some constants, with $m' \leqslant m$. Then consider the testing hypotheses
\begin{align}\label{eq:test_proc_individual}
	H_{0j}:\theta_{H_j}=c_j\quad &vs. \quad H_1:\theta_{H_j}\neq c_j, \quad\text{(individual test)},\\
	H_0:\theta_{H}=c\quad &vs.\quad H_1:\theta_{H}\neq c \quad\text{(multiple test)}.\label{eq:test_proc_multiple}
\end{align}
To obtain test statistics for tests in \eqref{eq:test_proc_individual} and \eqref{eq:test_proc_multiple}, we need to simply replace $\theta_j$ by $c_j$ in the definition of test statistic $t_j$ in \eqref{eq:statistics}, that is
\begin{equation*}\label{eq:statistics2}
	t_{H_j} = \frac{\hat{\theta}_{H_j}- c_j}{\hat{\sigma}_j},  \quad  
	t_{H_{0j}} = \frac{\hat{\theta}_{H_j}- \theta_{H_j}}{\hat{\sigma}_j}, \quad	M_{H} = \max_{j=1,\dots m'}\left\lvert  t_{H_j} \right\rvert, \quad M_{H_0} = \max_{j=1,\dots m'}\left\lvert  t_{H_{0j}} \right\rvert,
\end{equation*}
where $t_{H_0j}$ and $M_{H_0}$ are for retrieving the critical values. Tests using statistics $t_{H_j}$ and $M_{H}$ reject $H_{0j}$ and $H_0$ at the $\alpha$-level if $t_{H_j}\geq q_{H_{0j},\alpha}$ and $M_{H} \geq q_{H_{0}, \alpha}$ where $q_{H_{0j},\alpha}=\inf\{a\in \mathbb{R}: P(t_{H_j}\leq a)\geq 1 -\alpha\}$ and $q_{H_{0}, \alpha}=\inf\{a\in \mathbb{R}: P(M_{H_0} \leq a)\geq 1 -\alpha\}$.

Construction of the studentized statistics in \eqref{eq:statistics} requires the estimation of $\hat{\sigma}^2_j$. The most common measure to assess the variability of the prediction is the mean squared error $\mathrm{MSE}(\hat{\theta}_j)=E(\hat{\theta}_j-\theta_j)^2$, where $E$ denotes the expectation with respect to model \eqref{eq:LMM_b}. Nevertheless, following \cite{chatterjee2008parametric}, a simpler choice of $\sigma^2_j = l^t_j(G_j-G_jZ^t_jV_j^{-1}Z_jG_j)l_j$
which accounts for the variability of $\theta_j$ without accounting for 
the estimation of $\beta$ or $\delta$ lead to the most satisfactory numerical results. Simulation results showing finite sample performance of the intervals constructed using other variability estimators can be found in our Supplementary Material.

\section{Inference robust to misspecifications by semiparametric bootstrap}\label{sec:nonparametric_booststrap}
We present a bootstrap scheme to construct individual and simultaneous intervals which are robust to model misspecifications. Denote bootstrap generated observations by
\begin{equation}\label{eq:bootstrap_sample}
	y^*=X\hat{\beta}+Zu^*+e^*,
\end{equation}
where $e^*$ and $u^*$ are bootstrap replica of the random components in the model. We further set $\delta^*= \hat{\delta}$, $V^*=\hat{V}$, $G^* = \hat{G}$ and define
$\tilde{\beta}^* =\beta(\delta^*)= (X^tV^{*-1}X)^{-1}X^tV^{*-1}y^*$, $\tilde{u}^*_j=u_j(\delta^*)=G^*_j Z^t_jV_j^{*-1}(y^*_j-X_j \tilde{\beta}^*)$. In addition, let $\hat{\delta}^*$ be an estimated version of $\delta^*$ obtained by regressing  $y^*$ on $X$. Then we have $\hat{\beta}^* = \beta(\hat{\delta}^*)$ and  $\hat{u}^*_j= u_j(\hat{\delta}^*)$. Bootstrap mixed effects are thus defined as
\begin{align*}\label{eq:theta_boot}
	\theta^*_j = k^T_j\beta^*+l^T_ju^*_j, 	
	& &\tilde{\theta}^*_j = \theta_j(\delta^*) = k^T_j\tilde{\beta}^*+l^T_j\tilde{u}^*_j,
	& &\hat{\theta}^*_j = \theta_j(\hat{\delta}^*) = k^T_j\hat{\beta}^*+l^T_j\hat{u}^*_j.
\end{align*}
The bootstrap versions of the statistics of interest in \eqref{eq:statistics} are given by
\begin{equation}\label{eq:boot_statistics}
	t^*_j = \frac{\hat{\theta}^*_j- \theta^*_j}{\hat{\sigma}^*_j}, \quad M^* = \max_{j=1,\dots m}\left\lvert  t^*_j\right\rvert.
\end{equation}
We use statistics in \eqref{eq:boot_statistics} to construct bootstrap equivalents of intervals in \eqref{eq:ind_sim_interval}, that is
\begin{align}\label{eq:ind_interval_b}
	q^*_{j,\alpha}=\inf\{a\in \mathbb{R}: P(t^*_j\leq a)\geq 1 -\alpha\}, \;I^*_{j, \alpha}: \big\{\hat{\theta}_j  \pm q^*_{j,\alpha}\times \hat{\sigma}_j\big\}, \; j = 1, \dots, m,\\
	\label{eq:sim_interval_b}
	q^*_{\alpha}=\inf\{a\in \mathbb{R}: P(M^* \leq a)\geq 1 -\alpha\}, \; I^*_{\alpha} = \bigtimes_{j=1}^m 	I^{*s}_{j,\alpha}, \;
	I^{*s}_{j,\alpha} : \big\{\hat{\theta}_j  \pm q^*_{\alpha}\times \hat{\sigma}_j\big\}.
\end{align}
The most popular choice is to use a parametric bootstrap and draw $e^*$ and $u^*$ from a postulated normal distribution with estimated variance parameters. In contrast, we use a semiparametric bootstrap method introduced by \cite{carpenter2003novel} and generalised by \cite{opsomer2008splines}. The empirical performance of this bootstrap scheme for fixed parameters has been studied by \cite{chambers2013random}. The goal is to mimic the data generating process in model \eqref{eq:LMM_b}. Before writing down explicitly the bootstrap algorithm, we provide some motivation behind it. Let $\tilde{y}=X\tilde{\beta} = X(X^TV^{-1}X)^{-1}X^TV^{-1}y = Hy$, $\tilde{e} =  y -X\tilde{\beta} -Z\tilde{u} = (I - ZGZ^TV^{-1})(I-H)y =  RV^{-1}(I-H)y$ and $\hat{e} = y -X\hat{\beta} -Z\hat{u}$. Then, by some algebraic transformations we have $I - ZGZ^TV^{-1} = RV^{-1}$,  which leads to  $var(\tilde{u}) = GZ^T\{V^{-1}(I - H)\}ZG $ and $var(\tilde{e}) =  R\{V^{-1}(I-H)\}R$. Thus, we should re-scale $\hat{e}$ and $\hat{u}$ before sampling with replacement to avoid the effects of shrinkage \citep{morris2002blups}. Centring, that is subtracting the empirical mean, is also advisable to assure that the empirical re-scaled residuals have mean zero. This suggests sampling from $\hat{e}_{sc}$ and $\hat{u}_{sc}$ defined as follows
\begin{align*}\label{eq:e_scaled_centering}
	\hat{e}_{sc} &= \hat{e}_{s}-\bar{\hat{e}}_{s}, \quad \bar{\hat{e}}_{s} = \sum_{i=1}^{n} \frac{\hat{e}_{si}}{n}, \quad \hat{e}_{s} = [R\{V^{-1}(I-H)\}]^{-1/2}\hat{e},\\
	\hat{u}_{sc} &= \hat{u}_{s}-\bar{\hat{u}}_{s}, \quad  \bar{\hat{u}}_{s} = \sum_{i=1}^{n} \frac{\hat{u}_{sj}}{m}, \quad
	\hat{u}_s = [GZ^T\{V^{-1}(I - H)\}Z]^{-1/2}\hat{u}.
\end{align*}
The algorithm to obtain bootstrap quantiles and construct intervals in \eqref{eq:ind_interval_b} and \eqref{eq:sim_interval_b} is:

{\sl A semiparametric random effects bootstrap algorithm}

\begin{enumerate}[topsep=0pt]
	\setlength\itemsep{-1em}
\item Obtain consistent estimators $\hat{\beta}$ and $\hat{\delta}$.\\
\item For $b=1$ to $b=B$: 
  \begin{enumerate}[topsep=0pt]	\setlength\itemsep{-1em}
  \item Obtain vectors 
        $u^{*}\in \mathbb{R}^{m}$, $e^{*}\in \mathbb{R}^{n}$ by sampling independently with replacement from $\hat{u}_{sc}$ and $\hat{e}_{sc}$.\\
	\item  Generate sample	$y^{*}=X\hat{\beta}+Zu^{*(b)}+e^{*}$ in \eqref{eq:bootstrap_sample} and obtain $\theta^{*}_j$, $j = 1, \dots, m$.\\
   \item Fit LMM to bootstrap sample from the previous step.\\
	\item Obtain bootstrap estimates $\hat{\delta}^{*}$, $\hat{\beta}^{*}$, $\hat{\theta}_{j}^{*}$, $t^{*}_j$ and $M^{*}$, $j = 1, \dots, m$.\\
\end{enumerate}	
		\item Estimate critical values $q^*_{j,\alpha}$, $q^*_{\alpha}$ by the $[\{(1-\alpha)B\}+1]^{th}$ order statistics of $t^{*}_j$ and $M^*$, $j = 1, \dots, m$. \\
		\item Construct bootstrap intervals as indicated in \eqref{eq:ind_interval_b} and \eqref{eq:sim_interval_b}.
	\end{enumerate}
Fisher consistency of $\hat{\delta}^{*}$ and $\hat{\beta}^{*}$ obtained using semiparamteric bootstrap in the above algorithm has been proved by \cite{carpenter2003novel}. In Lemma \ref{lemma:consistency_t_boot} and \ref{lemma:consistency_M_boot} we show the consistency of statistics $t^*_j$ and $M^*$.

\begin{lemma}[Consistency of  $t^*_j$]\label{lemma:consistency_t_boot}
	Let $F_{t_j}(a) = P(t_j < a)$, $F_{t^*_j}(a) = P(t^*_j < a)$ be the cumulative distribution functions of statistics $t_j$, $t^*_j$ defined in \eqref{eq:statistics} and \eqref{eq:boot_statistics}. If the regularity conditions in Appendix 1 
	are satisfied, then we have in probability
	\begin{equation*}
		\sup_{a\in \mathbb{R}}\left\lvert F_{t_j}(a) - F_{t^*_j}(a) \right\rvert\to 0.
	\end{equation*}
\end{lemma}
	\begin{proof}
	Without loss of generality, we assume that the sequence of estimators $t_j$ converges to a continuous distribution function $F$. A standard way of proving the consistency of bootstrap procedure in Lemma \ref{lemma:consistency_t_boot} \citep[see, for example,][Chapter 23]{van2000asymptotic} is to show that, for every $a$ $F_{t_j}(a)\to F(a)$ in distribution and $F_{t^*_j}(a)\to F(a)$ 
	given the original sample size in probability. Let $\hat{\vartheta}^* = (\hat{\beta}^*, \hat{\delta}^*)$ and $E^*$ be a bootstrap operator of the expected value. Then $t_j$ and $t^*_j$ can be written as $t_j=f(\vartheta, \hat{\vartheta}, u_j)$ and $t^*_j=f(\hat{\vartheta}, \hat{\vartheta}^*, u^*_j)$, respectively for a continuous and a differentiable function  $f$. Consider a general score equation $s_n(\vartheta)$ defined in Appendix 
	and its bootstrap equivalent $s_n^*(\vartheta) =\sum_{j=1}^{m}\sum_{i = 1}^{n_j} \psi (y^*_{ij},\vartheta)$ with $y$ replaced by $y^*$. It follows that $E^*\{s_n^*(\vartheta)\} = 0$ at $\vartheta = \hat{\vartheta}$ which yields the consistency of the sequence of bootstrap estimators $\hat{\vartheta}^*$. The consistency of random effects under random effect bootstrap was proved by \cite{field2007bootstrapping} under Condition~$4$ in Appendix 
	which is in alignment with results of \cite{jiang1998asymptotic}. We thus have that $\sqrt n(\hat{\theta}^*_j- \theta^*_j)$ and $\sqrt n(\hat{\theta}_j- \theta_j)$ converge to the same distribution. Final consistency result follows by Slutsky's lemma. 
\end{proof}

Corollary \ref{cor:intervals} ensures the consistency of the individual confidence intervals. 

\begin{corollary}[Consistency of $I^*_{j,\alpha}$] \label{cor:intervals}
	Lemma \ref{lemma:consistency_t_boot} implies that under the same assumptions 
	\begin{equation*}
		P(\theta_j \in I^*_{j,\alpha}) \to 1-\alpha.
	\end{equation*}
\end{corollary}
\begin{proof}
	The proof follows along the same line as Lemma 23.3 in \cite{van2000asymptotic}. By Lemma \ref{lemma:consistency_t_boot}, the sequences of distribution functions $F_{t_j}$ and  $F_{t^*_j}$ converge weakly to $F$, which implies the convergence of their quantile functions $F^{-1}_{t_j}$ and  $F^{-1}_{t^*_j}$ at every continuity point. We thus conclude that $q^*_{j,\alpha}=F^{-1}_{t^*_j}(1-\alpha) \to F^{-1}(1-\alpha) $ almost surely, and
	\begin{equation*}
		P(\theta_j\geq \hat{\theta}_j - \hat{\sigma}_j q^*_{j,\alpha})= P\left( \frac{\hat{\theta}_j- \theta_j}{\hat{\sigma}_j} \leq  q^*_{j,\alpha} \right)\to P\left\{ t_j \leq F^{-1}(1-\alpha)\right\} = 1- \alpha
	\end{equation*}
	which completes the proof. 
\end{proof}
The consistency of $M^*$ does not follow from Lemma \ref{lemma:consistency_t_boot} by the delta method, because $\max$ function is not differentiable.  Instead, Lemma \ref{lemma:consistency_M_boot} provides a heuristic proof based on results known from the extreme value theory. 
\begin{lemma}[Consistency of  $M^*$]\label{lemma:consistency_M_boot}
	Let $M$ and $M^*$ be as defined in \eqref{eq:statistics} and  \eqref{eq:boot_statistics}. If the regularity conditions in Appendix 
	are satisfied and Lemma \ref{lemma:consistency_t_boot} holds, then we have in probability
	\begin{equation*}
		\sup_{a\in \mathbb{R}}\left\lvert F_{M}(a) - F_{M^*}(a) \right\rvert\to 0.
	\end{equation*}
\end{lemma}
\begin{proof}
	Observe that $F_{M}(a) = P(M < a)  = P (t_1\leq a, \dots, t_m\leq a, -t_1\leq a, \dots, -t_m\leq a)$.
	Since $t_j$, $j = 1, \dots, m$ are asymptotically independent and identically distributed, we have an approximation $F_M(a)\approx \prod_{j=1}^{2m}F_j(a)$
	with $F_j(a)$ some proper, non-degenerate distributions. By classical results in extreme value theory \citep{beirlant2004statistics,embrechts2013modelling}, we can assume that there exist sequences of re-normalizing constants $\{b_j > 0\}$, $\{c_j\}$ such that $P\{(M_{\theta} - c_j)/b_j \leq a\} $ converges to a  non-degenerate distribution function $H(a)$ as $j\to \infty$, i.e., the $F_j(a)$ belong to the max-domain of attraction of some non-degenerate, continuous distribution $H(a)$. The consistency of $F_{M^*}(a)$ follows by evoking the properties of the random effects bootstrap and the arguments used in the proof of Lemma \ref{lemma:consistency_t_boot}.
\end{proof}
\begin{corollary}\label{cor:sim_int}
	Lemma \ref{lemma:consistency_M_boot} implies that under the same assumptions 
	\begin{equation*}
		P(\theta_j \in I^*_{\alpha} \; \forall j\in[m]) \to 1-\alpha.
	\end{equation*}
\end{corollary}
\begin{proof}
	The proof follows now along the same lines as in Corollary \ref{cor:intervals} with statistic $t_j$ replaced by $M$.
\end{proof}
Similarly as in case of intervals, we can use semiparametric bootstrap to  approximate critical values $q_{H_{0j},\alpha}$ and $q_{H_{0}, \alpha}$ for tests in \eqref{eq:test_proc_individual} and \eqref{eq:test_proc_multiple}. Thanks to the relation between intervals and test, the consistency proof for intervals applies also for testing procedures with some changes of the notation \citep[cf.][]{reluga2021jasa}.

\section{Simulation study}\label{sec:simulations}

We carry out numerical simulation studies to evaluate finite sample properties of our bootstrap intervals.  In all scenarios we generate outcomes from a linear mixed effect model in \eqref{eq:LMM_b} with a fixed and a random intercept, and a uniformly distributed covariate, that is, we set $x_{ij1}=1$, $z_{ij} = 1$,  $x_{ij2}\sim U(0,1)$. We consider three types of sample sizes to mimic joint asymptotics: in setting 1 we have $m = 25$, $n_j =5$, in setting 2: $m = 50$, $n_j =10$, and in setting 3: $m = 75$, $n_j =15$. Furthermore, in each simulation, errors and random effects are drawn from one of the following distributions: standard normal, Student's t with 6 degrees of freedom, or chi-square with 5 degrees of freedom. The distributions are always centred to zero and re-scaled to variances $var(e_{ij})$ and $var(u_{j})$ 
which are indicated in Tables \ref{tab:CI}-\ref{tab:SPI}. 
We compare the performance of our individual and simultaneous intervals in \eqref{eq:ind_interval_b} at the $\alpha = 0.05$ level obtained using semiparmetric bootstrap, parametric bootstrap as in \cite{chatterjee2008parametric} and \cite{reluga2021simultaneous} as well as intervals constructed using large-sample asymptotic approximations, that is, with a $(1-\alpha/2)$ and  $(1-\alpha/2m)$ quantiles from normal distributions (the latter by Bonferroni correction). We employ following criteria to assess the performance of intervals: empirical coverage probability for individual and simultaneous intervals, that is, $\mathrm{Cov}_{ind} =1/mS\sum_{j=1}^{m}\sum_{s=1}^{S}\bm{1} \{ \theta^{(s)}_j\in I^{*(s)}_{j,\alpha} \}$ and $\mathrm{Cov}_{sim} =1/S\sum_{s=1}^{S}\bm{1} \{ \theta^{(s)}_j\in I^{*(s)}_{\alpha} \,\, \forall j \in [m] \}$; average widths of the intervals $\mathrm{Width}=1/mS\sum_{j=1}^{m}\sum_{s=1}^{S}\rho^{(s)}_j$; the variance of widths $\mathrm{VarWidth}=1/m(S-1) \sum_{j=1}^{m} \sum_{s=1}^{S} 
\left(\rho^{(s)}_j-\bar{\rho}_j \right)^2$, all of them over $S = 1000$ simulation runs, where $\rho^{(s)}_j=2q^{(s)}_{(\cdot)} \hat{\sigma}^{(s)}_j$, $\bar{\rho}_j=\sum_{s=1}^{S} \rho^{(s)}_j/S$ and $(\cdot)$ stands for the pair $j,\alpha$ for individual intervals and for $\alpha$ for simultaneous intervals.

Table \ref{tab:CI} displays the numerical performance of individual intervals for mixed effect $\theta_j$ in \eqref{eq:theta}. In this case, the performance of all methods seems to be similar -- the distribution of errors and random effects does hardly affect the empirical coverage, even for the intervals derived asymptotically. 
Our simulations indicate a surprisingly 
strong robustness to distributional misspecifications and the application of bootstrapping seems superfluous in this setting. 
\begin{table}[htb]
	\centering
	\def~{\hphantom{0}}
	{
		\begin{tabular}{|cccccccccccc|}\hline
			&          &   & \multicolumn{3}{c}{Coverage} & \multicolumn{3}{c}{Length} & \multicolumn{3}{c|}{Variance of length} \\
			$e_{ij}$         & $u_j$        & M   & $S_1$       & $S_2$      & $S_3$      & $S_1$      & $S_2$      & $S_3$     & $S_1$       & $S_2$      & $S_3$      \\\hline
			&          & A & 953      & 949     & 949     & 1558    & 1142    & 955    & 10       & 1       & 1       \\
			$N(1)$    & $N(0.5)$ & S & 954      & 947     & 947     & 1573    & 1138    & 951    & 13       & 3       & 3       \\
			&          & P & 962      & 948     & 947     & 1627    & 1137    & 950    & 27       & 3       & 3       \\
			&          & A & 946      & 948     & 949     & 1180    & 855     & 704    & 12       & 2       & 2       \\
			$t_6(0.5)$  & $t_6(1)$   & S & 947      & 947     & 949     & 1197    & 857     & 704    & 15       & 2       & 2       \\
			&          & P & 947      & 946     & 948     & 1191    & 852     & 701    & 14       & 2       & 2       \\
			&          & A & 948      & 950     & 950     & 1209    & 865     & 710    & 14       & 2       & 2       \\
			$\chi^2_5(0.5)$ & $\chi^2_5(1)$  & S & 945      & 945     & 945     & 1219    & 862     & 706    & 16       & 2       & 2       \\
			&          & P & 948      & 949     & 949     & 1218    & 864     & 707    & 15       & 2       & 2       \\
			&          & A & 948      & 951     & 952     & 1184    & 855     & 704    & 13       & 2       & 2       \\
			$\chi^2_5(0.5)$ & $t_6(1)$   & S & 946      & 946     & 946     & 1196    & 853     & 701    & 14       & 2       & 2       \\
			&          & P & 950      & 949     & 949     & 1193    & 852     & 702    & 13       & 2       & 2       \\
			&          & A & 947      & 949     & 949     & 1608    & 1183    & 980    & 21       & 3       & 3       \\
			$t_6(1)$    & $\chi^2_5(0.5)$ & S & 949      & 948     & 948     & 1632    & 1186    & 981    & 26       & 5       & 5       \\
			&          & P & 950      & 948     & 938     & 1641    & 1180    & 981   & 27       & 4       & 4   \\\hline   
	\end{tabular}}
	\caption{Empirical coverage, width and variance of widths of individual intervals at $\alpha = 0.05$ level. $S_1$, Setting 1; $S_2$, Setting 2, $S_3$, Setting 3; M, Method; A, asymptotic; S, semiparametric bootstrap; P, parametric bootstrap. All numerical entries are multiplied by 1000.  }
	\label{tab:CI}
\end{table}
\begin{table}[htb]
	\centering
	\def~{\hphantom{0}}{
		\begin{tabular}{|ccc ccc ccc ccc|}\hline
			&           &         &\multicolumn{3}{c}{Coverage} & \multicolumn{3}{c}{Length} & \multicolumn{3}{c|}{Variance of length} \\
			$e_{ij}$         & $u_j$         & M&  $S_1$    & $S_2$    & $S_3$   & $S_1$   & $S_2$   & $S_3$  & $S_1$    & $S_2$    & $S_3$   \\\hline
			&           & A   &    931     & 953      & 948     & 2456   & 1918    & 1658    & 25      & 4        & 1       \\
			$N(1)$    & $N(0.5)$  & S      & 937     & 955      & 952     & 2522   & 1923    & 1659    & 30      & 5        & 2       \\
			&           & P      & 971     & 958      & 951     & 2824   & 1925    & 1658    & 1356    & 4        & 2       \\
			&           & A      & 887     & 915      & 922     & 1860   & 1435    & 1222    & 30      & 4        & 1       \\
			$t_6(0.5)$  & $t_6(1)$    & S      & 924     & 945      & 952     & 1975   & 1509    & 1279    & 55      & 13       & 6       \\
			&           & P      & 900     & 919      & 918     & 1907   & 1439    & 1223    & 33      & 5        & 2       \\
			&           & A      & 886     & 866      & 905     & 1906   & 1452    & 1233    & 34      & 5        & 1       \\
			$\chi_5^2(0.5)$ & $\chi_5^2(1)$   & S      & 921     & 917      & 937     & 2041   & 1555    & 1316    & 57      & 10       & 3       \\
			&           & P      & 897     & 867      & 911     & 1948   & 1459    & 1234    & 36      & 5        & 2       \\
			&           & A   & 896     & 874      & 902     & 1867   & 1436    & 1222    & 31      & 4        & 1       \\
			$\chi^2_5(0.5)$ & $t_6(1)$    & S    & 932     & 924      & 938     & 2002   & 1537    & 1302    & 52      & 9        & 3       \\
			&           & P    & 913     & 884      & 905     & 1910   & 1440    & 1224    & 31      & 5        & 2       \\
			&           & A    & 899     & 898      & 914     & 2535   & 1986    & 1702    & 52      & 10       & 3       \\
			$t_6(1)$    & $\chi^2_5(0.5)$ & S   & 935     & 935      & 944     & 2694   & 2087    & 1779    & 90      & 29       & 13      \\
			&           & P   & 916     & 916      & 834     & 2657   & 1994    & 1544    & 149     & 9        & 251  \\\hline  
	\end{tabular}}
	\caption{Empirical coverage, width and variance of widths of simultaneous intervals at $\alpha = 0.05$ level. 	$S_1$, Setting 1; $S_2$, Setting 2, $S_3$, Setting 3; M, Method; A, asymptotic; S, semiparametric bootstrap; P, parametric bootstrap. All numerical entries are multiplied by 1000.  }
	\label{tab:SPI}
\end{table}
The situation changes dramatically in Table \ref{tab:SPI} which shows numerical performance of simultaneous intervals. In this case, the results are similar for all methods only when errors and random effects are normally distributed \citep[cf. results in][]{reluga2021simultaneous}. Regardless of the distribution of errors and/or random effects, the performance of intervals obtained using semiparametric bootstrap is superior to other methods. In fact, their application leads to serious undercoverage even for large sample sizes under departures from normality. 
Furthermore, the average length of semiparametric bootstrap intervals is not excessively wide in comparison to other methods. We can thus conclude that the application of our semiparametric bootstrap-based method leads to a satisfactory numerical performance even under considerable departures from normality. In comparison to other robust techniques, it does not involve robust estimation or any data transformation, which is extremely appealing for practitioners. 

\section{Discussion}\label{sec:conclusions}
Linear mixed effects are popular to predict cluster-level parameters in various domains. Yet, the underlying assumptions which should guarantee their satisfactory numerical performance are often violated in practice. We studied to what extent the application of a simple bootstrapping scheme might mitigate the negative effects of distributional misspecificaitons without the need to reach for more advanced techniques such as robust estimation or data transformation. Our numerical study confirms that mixed effects are fairly robust to such misspecification unless they undergo complex transformations. This is particularly interesting for their application in small area estimation in which mixed effects are often used in nonlinear poverty indicators \citep{RojasEtal2020} for which the application of semiparameric bootstrap inference could be particularly beneficial. 

\appendix
\section*{Appendix 1}
\subsection*{Regularity conditions}
We adopt some regularity conditions from \cite{shao2000consistency} and \cite{reluga2021simultaneous}. Let $\vartheta = (\beta, \delta)$, $\hat{\vartheta} = (\hat{\beta}, \hat{\delta})$ and $\vartheta_0\in \Theta \subset \mathbb{R}^{p+h+1}$ be the true parameter value. We assume that
\begin{enumerate}
	\item Score equation  $s_n(\vartheta) = \sum_{j=1}^{m}\sum_{i = 1}^{n_j} \psi (y_{ij},\vartheta)$ is well defined if:
	(a) $s_n(\vartheta)$ is continuous and differentiable for each fixed $y$, (b) $E\{s_n(\vartheta) \} = 0$ at $\vartheta_0$, (c) $\vartheta_0$ is an interior point of $\Theta$ and the estimator $\hat{\vartheta}$ is an interior point of the neighborhood of $\vartheta_0$.
	\item $\liminf\limits_{n}\lambda[n^{-1}\mathrm{var}\{s_n(\vartheta)\}]>0$ and  $\liminf\limits_{n}\lambda[-n^{-1}E \{\nabla s_n(\vartheta)\}]>0$ where $\nabla s_n(\vartheta) =\frac{\partial \psi(\vartheta)}{\partial \vartheta}$ and $\lambda[A]$ indicates the smallest eigenvalue of matrix $A$.
	\item There exists $b>0$ such that $E\norm{\psi(y_{ij}, \vartheta}^{2+b}<\infty$, and $E(h_N(y_{ij})^{1+b}$ in a compact neighbourhood $N$, where $h_C(y_{ij})=\sup_{\vartheta\in N}\norm{\nabla s_n(\vartheta)}$.
	\item Convergence: $m\to \infty$, $n_j\to \infty$.
	\item $V_j(\delta)$ has a linear structure in $\delta$, $j = 1, \dots, m$.
\end{enumerate}
Conditions 1--3 ensure that one can use the score equation $s_n$ to estimate fixed parameters $\vartheta$ up to a vanishing term. Condition $4$ is required to ensure the convergence of mixed effect predictors, whereas Condition $5$ implies that the second derivatives of $R_j$ and $G_j$ are 0. The assumption of $m\to\infty$, which is common in small area estimation literature once the modelling assumptions are satisfied \citep[cf.][in the context of simultaneous inference]{reluga2021simultaneous}, must be replaced by the joint asymptotics in Condition 4 to ensure the convergence of cumulative distributions functions of mixed effects under departures from normality \citep[cf.][]{jiang1998asymptotic}. Nevertheless, this assumption is important only for the theoretical derivations -- in practice bootstrap intervals perform well for a sample size as small as $n_j = 5$ (cf., results in Tables \ref{tab:CI}-\ref{tab:SPI}).

\section*{Appendix 2}
\subsection*{Additional simulation results}\label{sec:additional_simulations}	
	In this section, we present additional simulations results using different MSE estimators. 
	Analytical MSE can be decomposed as follows 
	\begin{align}
		\mathrm{MSE}(\hat{\theta}_j )&=\mathrm{MSE}(\tilde{\theta}_j )+E\left(\hat{\theta}_j 
		-\tilde{\theta}_j \right)^2+2E\left\{(\tilde{\theta}_j -\theta_j)(\hat{\theta}_j -\tilde{\theta}_j )\right\}\nonumber\\
		&=g_{1j}({\delta}) + g_{2j}(\delta) +  g_{3j}(\delta) +2E\left\{(\tilde{\theta}_j -\theta_j)(\hat{\theta}_j -\tilde{\theta}_j )\right\}\label{eq:MSE},
	\end{align}
	where $\mathrm{MSE}(\tilde{\theta}_j)$ accounts for the variability of $\theta_j$ when the variance components $\delta$ are known. It particular, $g_{1j}$ accounts for the variability of $\theta_j$ for known $\beta$, $g_{2j}$ for the estimation of $\beta$, $g_{3j}$ quantifies the square difference between $\hat{\theta}_j$ and $\tilde{\theta}_j$. There exists a vast literature to estimate it \citep[see, for example,][]{rao2015small}. The last term in \eqref{eq:MSE} disappears under normality of errors and random effects. 
	Let $b_j^T=k^T_j-o_j^TX_j$ with $o^T_j=l^t_jG_jZ^t_jV^{-1}_j$. Under linear mixed model, the analytical estimator of variability $\mathrm{mse}_{L}(\hat{\theta}_j )$ reduces to
	\begin{equation*}\label{eq:MSE_lin}
		\mathrm{mse}_{L}(\hat{\theta}_j )= g_{1j}(\hat{{\delta}}) + g_{2j}(\hat{\delta}) + 2 g_{3j}(\hat{\delta}),
	\end{equation*}
	and $g_1$, $g_2$ and $g_3$ are defined in expression \eqref{eq:MSE_b_first_term}:
	\begin{equation}\label{eq:MSE_b_first_term}
		\begin{split}
			g_{1j}(\delta)&=l^t_j(G_j-G_jZ^t_jV_j^{-1}Z_jG_j)l_j,\\ 
			g_{2j}(\delta)&=b_j^t\left(\sum_{j=1}^{m}X^t_{j}V^{-1}_jX_j\right)^{-1}b_j,\\
			g_{3j}(\delta)&=\text{tr}\left\{ (\partial o_j^t/\partial\delta) V_j (\partial o_j^t/\partial\delta)^tV_A(\hat{\delta}) \right\},
		\end{split}
	\end{equation}
	where $V_A(\hat{\delta})$ the asymptotic covariance matrix. In addition, $E\left\{\mathrm{mse}_{L}(\hat{\theta}_j )\right\}=\mathrm{MSE}(\theta_j )+o(m^{-1})$. First, we complete the numerical results from Section \ref{sec:simulations} by considering additional simulation scenarios. Tables \ref{tab:CI_g1}-\ref{tab:SPI_g1} show the numerical results with $\hat{\sigma}^2_j = g_{1j}$.
	\begin{table}[htb]
		\centering
		\def~{\hphantom{0}}
		{
			\begin{tabular}{|cccccccccccc|}\hline
				&           &   & \multicolumn{3}{c}{Coverage} & \multicolumn{3}{c}{Length} & \multicolumn{3}{c|}{Variance of length} \\
				$e_{ij}$        & $u_j$         & M & $S_1$       & $S_2$      & $S_3$      & $S_1$      & $S_2$      & $S_3$     & $S_1$       & $S_2$      & $S_3$      \\\hline
				&           & A & 948      & 949     & 949     & 1191    & 856     & 705    & 6        & 1       & 1       \\
				$N(0.5)$ & $N(1)$    & S & 948      & 948     & 947     & 1201    & 855     & 702    & 7        & 1       & 1       \\
				&           & P & 949      & 948     & 947     & 1202    & 855     & 702    & 7        & 1       & 1       \\
				&           & A & 952      & 949     & 950     & 1529    & 1137    & 952    & 16       & 3       & 3       \\
				$t_6(1)$   & $t_5(0.5)$  & S & 953      & 948     & 949     & 1552    & 1138    & 951    & 20       & 4       & 4       \\
				&           & P & 963      & 947     & 951     & 1627    & 1133    & 948    & 62       & 4       & 4       \\
				&           & A & 949      & 949     & 950     & 1614    & 1183    & 981    & 22       & 3       & 3       \\
				$\chi_5(1)$  & $\chi_5(0.5)$ & S & 949      & 946     & 946     & 1633    & 1181    & 977    & 26       & 4       & 4       \\
				&           & P & 951      & 951     & 949     & 1647    & 1182    & 978    & 27       & 4       & 4       \\
				&           & A & 949      & 949     & 950     & 1614    & 1183    & 981    & 22       & 3       & 3       \\
				$t_6(0.5)$ & $\chi_5(1)$   & S & 949      & 946     & 946     & 1633    & 1181    & 977    & 26       & 4       & 4       \\
				&           & P & 951      & 951     & 948     & 1647    & 1182    & 707    & 27       & 4       & 4    \\  \hline
		\end{tabular}}
		\caption{Empirical coverage, width and variance of widths of individual intervals at $\alpha = 0.05$-level, $\sigma^2_j = g_{1j}$. $S_1$, Setting 1; $S_2$, Setting 2, $S_3$, Setting 3; M, Method; A, asymptotic; S, semiparametric bootstrap; P, parametric bootstrap. All numerical entries are multiplied by 1000.}
		\label{tab:CI_g1}
	\end{table}
	
	\begin{table}[htb]
		\centering
		\def~{\hphantom{0}}
		{
			\begin{tabular}{|cccccccccccc|}\hline
				&           &   & \multicolumn{3}{c}{Coverage} & \multicolumn{3}{c}{Length} & \multicolumn{3}{c|}{Variance of length} \\
				$e_{ij}$         & $u_j$         & M & $S_1$       & $S_2$      & $S_3$      & $S_1$      & $S_2$      & $S_3$     & $S_1$        & $S_2$      & $S_3$     \\\hline
				&           & A & 932      & 955     & 944     & 1879    & 1438    & 1224   & 14        & 2       & 1      \\
				$N(0.5)$ & $N(1)$    & S & 945      & 956     & 947     & 1919    & 1443    & 1225   & 17        & 2       & 1      \\
				&           & P & 946      & 961     & 946     & 1924    & 1445    & 1225   & 16        & 2       & 1      \\
				&           & A & 916      & 926     & 920     & 2411    & 1909    & 1653   & 40        & 8       & 3      \\
				$t_6(1)$   & $t_6(0.5)$  & S & 949      & 951     & 951     & 2557    & 2002    & 1725   & 66        & 21      & 10     \\
				&           & P & 946      & 932     & 921     & 2823    & 1918    & 1655   & 1790      & 8       & 3      \\
				&           & A & 911      & 884     & 902     & 2544    & 1985    & 1703   & 55        & 8       & 3      \\
				$\chi^2_5(1)$  & $\chi^2_5(0.5)$ & S & 930      & 922     & 934     & 2703    & 2102    & 1802   & 87        & 16      & 5      \\
				&           & P & 929      & 886     & 903     & 2668    & 1996    & 1707   & 101       & 9       & 3      \\
				&           & A & 911      & 884     & 902     & 2544    & 1985    & 1703   & 55        & 8       & 3      \\
				$t_6(0.5)$ & $\chi_5(1)$   & S & 930      & 922     & 934     & 2703    & 2102    & 1802   & 87        & 16      & 5      \\
				&           & P & 929      & 886     & 922     & 2668    & 1996    & 1233   & 101       & 9       & 2   \\  \hline
		\end{tabular}}
		\caption{Empirical coverage, width and variance of widths of simultaneous intervals at $\alpha = 0.05$-level, $\sigma^2_j = g_{1j}$. $S_1$, Setting 1; $S_2$, Setting 2, $S_3$, Setting 3; M, Method; A, asymptotic; S, semiparametric bootstrap; P, parametric bootstrap. All numerical entries are multiplied by 1000. }
		\label{tab:SPI_g1}
	\end{table}
	
	Alternatively, one could estimate MSE using bootstrap. The most straightforward bootstrap estimator is $\mathrm{MSE}^*(\hat{\theta}^*_j) = E^*\left(\hat{\theta}^*_j - \theta^*_j\right)^2 $ which might be approximated by  
	\begin{eqnarray}\label{eq:MSE_boot2}
		\mathrm{MSE}^*_{B1}(\hat{\theta}^*_j)\approx \mathrm{mse}^*_{B2}(\hat{\theta}_j ) = \frac{1}{B}\sum_{b=1}^{B} \left(\hat{\theta}^{*(b)}_j-\theta^{*(b)}_j\right)^2   , 
	\end{eqnarray}
	and $\hat{\theta}^{*(b)}_j$, $\theta^{*(b)}_j$ as defined in Section \ref{sec:nonparametric_booststrap}, calculated from the $b^{th}$ bootstrap sample. Tables \ref{tab:CI_b1}-\ref{tab:SPI_b1} display the performance of individual and simultaneous intervals constructed using $\mathrm{MSE}^*_{B1}$. As we can see, a general trend is the same as in case of $\sigma^2_j =g_{1j}$, that is there is not much different between the performance of parametric and semiparametric bootstrap individual intervals, but this changes dramatically if we consider simultaneous intervals. 
	
	\begin{table}[htb]
		\centering
		\def~{\hphantom{0}}
		{
			\begin{tabular}{|cccccccccccc|}\hline
				$e_{ij}$        & $u_j$         &   & \multicolumn{3}{c}{Coverage} & \multicolumn{3}{c}{Length} & \multicolumn{3}{c|}{Variance of length} \\
				&               &   & $S_1$    & $S_2$   & $S_3$   & $S_1$   & $S_2$   & $S_3$  & $S_1$    & $S_2$   & $S_3$   \\\hline
				&               & S & 943      & 947     & 947     & 1181    & 852     & 701    & 7        & 1       & 1       \\
				$N(0.5)$      & $N(1)$        & P & 944      & 947     & 946     & 1181    & 852     & 701    & 7        & 1       & 1       \\
				&               & S & 946      & 946     & 947     & 1521    & 1129    & 948    & 12       & 3       & 1       \\
				$N(1)$        & $N(0.5)$      & P & 939      & 946     & 946     & 1488    & 1128    & 947    & 15       & 3       & 1       \\
				&               & S & 943      & 946     & 948     & 1175    & 853     & 702    & 14       & 2       & 1       \\
				$t_6(0.5)$    & $t_6(1)$      & P & 943      & 945     & 947     & 1170    & 849     & 700    & 13       & 2       & 1       \\
				&               & S & 946      & 945     & 948     & 1501    & 1126    & 947    & 18       & 4       & 2       \\
				$t_6(1)$      & $t_6(0.5)$    & P & 954      & 949     & 950     & 1560    & 1137    & 952    & 21       & 3       & 1       \\
				&               & S & 941      & 944     & 945     & 1200    & 859     & 705    & 15       & 2       & 1       \\
				$\chi_5(0.5)$ & $\chi_5((1)$  & P & 945      & 948     & 949     & 1201    & 861     & 706    & 15       & 2       & 1       \\
				&               & S & 943      & 944     & 945     & 1590    & 1173    & 974    & 24       & 4       & 2       \\
				$\chi_5(1)$   & $\chi_5(0.5)$ & P & 943      & 948     & 948     & 1591    & 1176    & 976    & 24       & 4       & 2       \\
				&               & S & 943      & 944     & 945     & 1590    & 1173    & 974    & 24       & 4       & 2       \\
				$t_6(0.5)$    & $\chi_5(1)$   & P & 943      & 948     & 947     & 1591    & 1176    & 706    & 24       & 4       & 1       \\
				&               & S & 940      & 944     & 945     & 1173    & 849     & 699    & 14       & 2       & 1       \\
				$\chi_5(0.5)$ & $t_6(1)$      & P & 946      & 949     & 948     & 1172    & 849     & 701    & 13       & 2       & 1       \\
				&               & S & 943      & 946     & 947     & 1591    & 1179    & 978    & 24       & 5       & 2       \\
				$t_6(1)$      & $\chi_5(0.5)$ & P & 943      & 947     & 826     & 1588    & 1174    & 1534   & 24       & 4       & 2  \\  \hline
		\end{tabular}}
		\caption{Empirical coverage, width and variance of widths of individual intervals at $\alpha = 0.05$-level, $\sigma^2_j = MSE^*_{B1}(\hat{\theta}^*_j)$. $S_1$, Setting 1; $S_2$, Setting 2, $S_3$, Setting 3; M, Method; S, semiparametric bootstrap; P, parametric bootstrap. All numerical entries are multiplied by 1000.}
		\label{tab:CI_b1}
	\end{table}

	\begin{table}[htb]
		\centering
		\def~{\hphantom{0}}
		{
			\begin{tabular}{|cccccccccccc|}\hline
				$e_{ij}$        & $u_j$         &   & \multicolumn{3}{c}{Coverage} & \multicolumn{3}{c}{Length} & \multicolumn{3}{c|}{Variance of length} \\
				&               &   & $S_1$    & $S_2$   & $S_3$   & $S_1$   & $S_2$   & $S_3$  & $S_1$    & $S_2$   & $S_3$   \\\hline
				$N(0.5)$      & $N(1)$        & S & 930      & 945     & 944     & 1860    & 1430    & 1219   & 17       & 3       & 2       \\
				&               & P & 928      & 948     & 939     & 1865    & 1432    & 1218   & 16       & 3       & 2       \\
				$N(1)$        & $N(0.5)$      & S & 922      & 947     & 946     & 2385    & 1895    & 1646   & 28       & 7       & 3       \\
				&               & P & 908      & 942     & 942     & 2355    & 1896    & 1645   & 36       & 6       & 3       \\
				$t_6(0.5)$    & $t_6(1)$      & S & 906      & 940     & 944     & 1922    & 1502    & 1276   & 58       & 17      & 8       \\
				&               & P & 875      & 919     & 909     & 1846    & 1426    & 1217   & 32       & 6       & 2       \\
				$t_6(1)$      & $t_6(0.5)$    & S & 923      & 938     & 938     & 2428    & 1972    & 1715   & 63       & 25      & 12      \\
				&               & P & 909      & 926     & 921     & 2460    & 1909    & 1652   & 53       & 8       & 3       \\
				$\chi_5(0.5)$ & $\chi_5((1)$  & S & 908      & 918     & 937     & 2009    & 1555    & 1316   & 64       & 12      & 4       \\
				&               & P & 884      & 856     & 898     & 1894    & 1448    & 1227   & 36       & 6       & 2       \\
				$\chi_5(1)$   & $\chi_5(0.5)$ & S & 921      & 922     & 931     & 2616    & 2094    & 1799   & 90       & 19      & 7       \\
				&               & P & 898      & 870     & 893     & 2513    & 1975    & 1696   & 56       & 10      & 4       \\
				$t_6(0.5)$    & $\chi_5(1)$   & S & 921      & 922     & 931     & 2616    & 2094    & 1799   & 90       & 19      & 7       \\
				&               & P & 898      & 870     & 914     & 2513    & 1975    & 1227   & 56       & 10      & 2       \\
				$\chi_5(0.5)$ & $t_6(1)$      & S & 919      & 918     & 936     & 1962    & 1534    & 1301   & 56       & 11      & 4       \\
				&               & P & 898      & 876     & 899     & 1850    & 1427    & 1217   & 31       & 6       & 2       \\
				$t_6(1)$      & $\chi_5(0.5)$ & S & 918      & 931     & 938     & 2597    & 2071    & 1772   & 94       & 34      & 16      \\
				&               & P & 871      & 906     & 825     & 2510    & 1973    & 1535   & 59       & 11      & 249    \\\hline
		\end{tabular}}
		\caption{Empirical coverage, width and variance of widths of simultaneous intervals at $\alpha = 0.05$-level, $\sigma^2_j = MSE^*_{B1}(\hat{\theta}^*_j)$. 	$S_1$, Setting 1; $S_2$, Setting 2, $S_3$, Setting 3; M, Method; S, semiparametric bootstrap; P, parametric bootstrap. All numerical entries are multiplied by 1000.}
		\label{tab:SPI_b1}
	\end{table}
	
	We can define several other bootstrap estimators. For example, $\mathrm{MSE}^*_{3T}$ directly approximates each term in \eqref{eq:MSE} by bootstrap, that is
	\begin{equation}\label{eq:MSE_3t}
		\mathrm{MSE}^*_{3T}(\hat{\theta}^*_j)=\mathrm{MSE}_B^*(\tilde{\theta}^*_j)+
		E^*(\hat{\theta}^*_j-\tilde{\theta}^*_j)^2+2E^*\left\{(\tilde{\theta}^*_j-\theta^*_j)(\hat{\theta}^*_j-\tilde{\theta}^*_j)\right\}.
	\end{equation}
	Tables \ref{tab:CI_3t}-\ref{tab:SPI_3t} display the performance of individual and simultaneous intervals constructed using $MSE^*_{3T}$. 
	
	\begin{table}[htb]
		\centering
		\def~{\hphantom{0}}
		{
			\begin{tabular}{|cccccccccccc|}\hline
				$e_{ij}$        & $u_j$         &   & \multicolumn{3}{c}{Coverage} & \multicolumn{3}{c}{Length} & \multicolumn{3}{c|}{Variance of length} \\
				&               &   & $S_1$    & $S_2$   & $S_3$   & $S_1$   & $S_2$   & $S_3$  & $S_1$    & $S_2$   & $S_3$   \\\hline
			$N(0.5)$      & $N(1)$            & S & 943      & 947     & 947     & 1181    & 852     & 701    & 7        & 1       & 1       \\
			   &         & P & 944      & 947     & 946     & 1181    & 852     & 701    & 7        & 1       & 1       \\
				$N(1)$        & $N(0.5)$          & S & 946      & 946     & 947     & 1521    & 1129    & 948    & 12       & 3       & 1       \\
				  &      & P & 939      & 946     & 946     & 1488    & 1128    & 947    & 15       & 3       & 1       \\
				$t_6(0.5)$    & $t_6(1)$               & S & 943      & 946     & 948     & 1175    & 853     & 702    & 14       & 2       & 1       \\
				  &      & P & 943      & 945     & 947     & 1170    & 849     & 700    & 13       & 2       & 1       \\
				$t_6(1)$      & $t_6(0.5)$      & S & 946      & 945     & 948     & 1501    & 1126    & 947    & 18       & 4       & 2       \\
				 &     & P & 954      & 949     & 950     & 1560    & 1137    & 952    & 21       & 3       & 1       \\
				$\chi_5(0.5)$ & $\chi_5(1)$             & S & 941      & 944     & 945     & 1200    & 859     & 705    & 15       & 2       & 1       \\
				&  & P & 945      & 948     & 949     & 1201    & 861     & 706    & 15       & 2       & 1       \\
				$\chi_5(1)$   & $\chi_5(0.5)$          & S & 943      & 944     & 945     & 1590    & 1173    & 974    & 24       & 4       & 2       \\
				& & P & 943      & 948     & 948     & 1591    & 1176    & 976    & 24       & 4       & 2       \\
				$t_6(0.5)$    & $\chi_5(1)$            & S & 943      & 944     & 945     & 1590    & 1173    & 974    & 24       & 4       & 2       \\
				&   & P & 943      & 948     & 947     & 1591    & 1176    & 706    & 24       & 4       & 1       \\
				$\chi_5(0.5)$ & $t_6(1)$              & S & 940      & 944     & 945     & 1173    & 849     & 699    & 14       & 2       & 1       \\
				&      & P & 946      & 949     & 948     & 1172    & 849     & 701    & 13       & 2       & 1       \\
				$t_6(1)$      & $\chi_5(0.5)$               & S & 943      & 946     & 947     & 1591    & 1179    & 978    & 24       & 5       & 2       \\
				& & P & 943      & 947     & 826     & 1588    & 1174    & 1534   & 24       & 4       & 249   \\\hline 
		\end{tabular}}
		\caption{Empirical coverage, width and variance of widths of individual intervals at $\alpha = 0.05$-level, $\sigma^2_j = MSE^*_{3T}(\hat{\theta}^*_j)$. $S_1$, Setting 1; $S_2$, Setting 2, $S_3$, Setting 3; M, Method; S, semiparametric bootstrap; P, parametric bootstrap. All numerical entries are multiplied by 1000.}
		\label{tab:CI_3t}
	\end{table}

	\begin{table}[htb]
		\centering
		\def~{\hphantom{0}}
		{
		\begin{tabular}{|cccccccccccc|}\hline
			$e_{ij}$        & $u_j$         &   & \multicolumn{3}{c}{Coverage} & \multicolumn{3}{c}{Length} & \multicolumn{3}{c|}{Variance of length} \\
			&               &   & $S_1$    & $S_2$   & $S_3$   & $S_1$   & $S_2$   & $S_3$  & $S_1$    & $S_2$   & $S_3$   \\\hline
			$N(0.5)$ & $N(1)$   & S      & 930      & 947   & 941   & 1861   & 1430  & 1219  & 17       & 3     & 2     \\
			&          & P      & 928      & 948   & 938   & 1865   & 1432  & 1218  & 16       & 3     & 2     \\
			$N(1)$   & $N(0.5)$ & S      & 922      & 946   & 945   & 2386   & 1895  & 1646  & 28       & 7     & 3     \\
			&          & P      & 912      & 940   & 944   & 2357   & 1896  & 1645  & 35       & 6     & 3     \\
			$t_6(0.5)$ & $t_6(1)$   & S      & 906      & 940   & 944   & 1923   & 1502  & 1276  & 58       & 17    & 8     \\
			&          & P      & 874      & 918   & 909   & 1846   & 1427  & 1217  & 32       & 6     & 2     \\
			$t_6(1)$   & $t_6(0.5)$ & S      & 926      & 940   & 938   & 2430   & 1973  & 1715  & 63       & 25    & 12    \\
			&          & P      & 862      & 916   & 913   & 2323   & 1888  & 1643  & 60       & 10    & 4     \\
			$\chi_5(0.5)$   & $\chi_5(1)$     & S      & 910      & 915   & 937   & 2009   & 1555  & 1316  & 63       & 12    & 4     \\
			&          & P      & 885      & 855   & 897   & 1895   & 1448  & 1227  & 36       & 6     & 2     \\
			$\chi_5(1)$    & $\chi_5(0.5)$   & S      & 919      & 922   & 931   & 2618   & 2094  & 1799  & 90       & 19    & 7     \\
			&          & P      & 899      & 868   & 892   & 2515   & 1975  & 1696  & 57       & 10    & 4     \\
			 $t_6(0.5)$    &  $\chi_5(1)$    & S      & 919      & 922   & 931   & 2618   & 2094  & 1799  & 90       & 19    & 7     \\
			&          & P      & 899      & 868   & 914   & 2515   & 1975  & 1227  & 57       & 10    & 2     \\
		    $\chi_5(0.5)$ & $t_6(1)$      & S      & 920      & 918   & 936   & 1963   & 1534  & 1301  & 56       & 11    & 4     \\
			&          & P      & 898      & 874   & 899   & 1850   & 1427  & 1217  & 31       & 6     & 2     \\
			$t_6(1)$     & $\chi_5(0.5)$  & S      & 916      & 931   & 940   & 2599   & 2071  & 1772  & 94       & 34    & 16    \\
			&          & P      & 874      & 907   & 824   & 2511   & 1973  & 1535  & 59       & 11    & 249 \\\hline 
	\end{tabular}}
		\caption{Empirical coverage, width and variance of widths of simultaneous intervals at $\alpha = 0.05$-level, $\sigma^2_j = MSE^*_{3T}(\hat{\theta}^*_j)$. $S_1$, Setting 1; $S_2$, Setting 2, $S_3$, Setting 3; M, Method; S, semiparametric bootstrap; P, parametric bootstrap. All numerical entries are multiplied by 1000. }
		\label{tab:SPI_3t}
	\end{table}
	
	It is well known that $\mathrm{MSE}^*_{3T}$ leads to estimators with bias of order $O(m^{-1})$. To obtain a bias of order $o(m^{-1})$, \cite{butar2003measures} advocate approximating only intractable terms in \eqref{eq:MSE} by bootstrap. Specifically, with 
	$g_{1d}( \cdot )$ and $g_{2d}(\cdot )$ as defined in (\ref{eq:MSE_b_first_term}), one takes
	\begin{eqnarray}\nonumber
		\mathrm{MSE}^*_{SPA}(\hat{\theta}^*_j) &=& 2 \left\{g_{1j}(\hat{\delta}) + g_{2j}(\hat{\delta})\right\} 
		- E^*\left\{g_{1j}(\hat{\delta}^*) + g_{2j}(\hat{\delta}^*)\right\} + E^*\left(\hat{\theta}^*_j-\tilde{\theta}^*_j\right)^2 \\ &
		+& 2E^*\left\{(\tilde{\theta}^*_j-\theta^*_j)(\hat{\theta}^*_j-\tilde{\theta}^*_j)\right\}, \label{eq:MSE_spa}
	\end{eqnarray}
	where the last term is zero under normality. Tables \ref{tab:CI_spa}-\ref{tab:SPI_spa} display the performance of individual and simultaneous intervals constructed using $MSE^*_{SPA}$. 
	\begin{table}[htb]
		\centering
		\def~{\hphantom{0}}{
		\begin{tabular}{|cccccccccccc|}\hline
				$e_{ij}$        & $u_j$         &   & \multicolumn{3}{c}{Coverage} & \multicolumn{3}{c}{Length} & \multicolumn{3}{c|}{Variance of length} \\
				&               & M & $S_1$    & $S_2$ & $S_3$ & $S_1$  & $S_2$ & $S_3$ & $S_1$    & $S_2$ & $S_3$ \\\hline
				$N(0.5)$      & $N(1)$        & S & 943      & 947   & 947   & 1181   & 852   & 701   & 7        & 1     & 1     \\
				&               & P & 944      & 947   & 946   & 1181   & 852   & 701   & 7        & 1     & 1     \\
				$N(1)$        & $N(0.5)$      & S & 946      & 946   & 947   & 1521   & 1129  & 948   & 12       & 3     & 1     \\
				&               & P & 939      & 946   & 946   & 1488   & 1128  & 947   & 15       & 3     & 1     \\
				$t_6(0.5)$    & $t_6(1)$      & S & 943      & 946   & 948   & 1175   & 853   & 702   & 14       & 2     & 1     \\
				&               & P & 943      & 945   & 947   & 1170   & 849   & 700   & 13       & 2     & 1     \\
				$t_6(1)$      & $t_6(0.5)$    & S & 946      & 945   & 948   & 1501   & 1126  & 947   & 18       & 4     & 2     \\
				&               & P & 937      & 945   & 947   & 1467   & 1123  & 945   & 26       & 4     & 2     \\
				$\chi_5(0.5)$ & $\chi_5((1)$  & S & 941      & 944   & 945   & 1200   & 859   & 705   & 15       & 2     & 1     \\
				&               & P & 945      & 948   & 949   & 1201   & 861   & 706   & 15       & 2     & 1     \\
				$\chi_5(1)$   & $\chi_5(0.5)$ & S & 943      & 944   & 945   & 1590   & 1173  & 974   & 24       & 4     & 2     \\
				&               & P & 943      & 948   & 948   & 1591   & 1176  & 976   & 24       & 4     & 2     \\
				$t_6(0.5)$    & $\chi_5(1)$   & S & 943      & 944   & 945   & 1590   & 1173  & 974   & 24       & 4     & 2     \\
				&               & P & 943      & 948   & 947   & 1591   & 1176  & 706   & 24       & 4     & 1     \\
				$\chi_5(0.5)$ & $t_6(1)$      & S & 940      & 944   & 945   & 1173   & 849   & 699   & 14       & 2     & 1     \\
				&               & P & 946      & 949   & 948   & 1172   & 849   & 701   & 13       & 2     & 1     \\
				$t_6(1)$      & $\chi_5(0.5)$ & S & 943      & 946   & 947   & 1591   & 1179  & 978   & 24       & 5     & 2     \\
				&               & P & 943      & 947   & 830   & 1588   & 1174  & 1538  & 24       & 4     & 248 \\ \hline
		\end{tabular}}
		\caption{Empirical coverage, width and variance of widths of individual intervals at $\alpha = 0.05$-level, $\sigma^2 = MSE^*_{SPA}$. $S_1$, Setting 1; $S_2$, Setting 2, $S_3$, Setting 3; M, Method; S, semiparametric bootstrap; P, parametric bootstrap. All numerical entries are multiplied by 1000. }
		\label{tab:CI_spa}
	\end{table}
	\begin{table}[htb]
		\def~{\hphantom{0}}
		\centering
		{\begin{tabular}{|cccccccccccc|}\hline
				$e_{ij}$        & $u_j$         &   & \multicolumn{3}{c}{Coverage} & \multicolumn{3}{c}{Length} & \multicolumn{3}{c|}{Variance of length} \\
				&               & M & $S_1$    & $S_2$   & $S_3$   & $S_1$   & $S_2$   & $S_3$  & $S_1$    & $S_2$   & $S_3$   \\\hline
				$N(0.5)$      & $N(1)$        & S & 934      & 956     & 946     & 1864    & 1433    & 1221   & 16       & 2       & 1       \\
				&               & P & 930      & 951     & 942     & 1868    & 1435    & 1221   & 15       & 2       & 1       \\
				$N(1)$        & $N(0.5)$      & S & 924      & 945     & 945     & 2390    & 1899    & 1650   & 26       & 5       & 2       \\
				&               & P & 910      & 947     & 947     & 2360    & 1900    & 1649   & 33       & 5       & 2       \\
				$t_6(0.5)$    & $t_6(1)$      & S & 908      & 944     & 949     & 1927    & 1506    & 1280   & 56       & 16      & 7       \\
				&               & P & 877      & 918     & 918     & 1850    & 1430    & 1220   & 30       & 5       & 2       \\
				$t_6(1)$      & $t_6(0.5)$    & S & 924      & 947     & 949     & 2436    & 1978    & 1719   & 61       & 24      & 11      \\
				&               & P & 860      & 921     & 919     & 2327    & 1892    & 1646   & 58       & 9       & 3       \\
				$\chi_5(0.5)$ & $\chi_5((1)$  & S & 910      & 917     & 938     & 2014    & 1560    & 1320   & 62       & 11      & 3       \\
				&               & P & 886      & 859     & 907     & 1898    & 1451    & 1230   & 34       & 5       & 2       \\
				$\chi_5(1)$   & $\chi_5(0.5)$ & S & 921      & 918     & 935     & 2623    & 2100    & 1805   & 87       & 17      & 6       \\
				&               & P & 898      & 878     & 901     & 2519    & 1980    & 1700   & 54       & 8       & 3       \\
				$t_6(0.5)$    & $\chi_5(1)$   & S & 921      & 918     & 935     & 2623    & 2100    & 1805   & 87       & 17      & 6       \\
				&               & P & 898      & 878     & 919     & 2519    & 1980    & 1230   & 54       & 8       & 2       \\
				$\chi_5(0.5)$ & $t_6(1)$      & S & 921      & 919     & 939     & 1968    & 1539    & 1306   & 54       & 10      & 3       \\
				&               & P & 897      & 878     & 904     & 1854    & 1430    & 1220   & 29       & 5       & 2       \\
				$t_6(1)$      & $\chi_5(0.5)$ & S & 914      & 933     & 944     & 2604    & 2078    & 1778   & 92       & 32      & 15      \\
				&               & P & 880      & 902     & 833     & 2515    & 1977    & 1538   & 56       & 9       & 249 \\   \hline
		\end{tabular}}
		\caption{Empirical coverage, width and variance of widths of simultaneous intervals at $\alpha = 0.05$-level, $\sigma^2 = MSE^*_{SPA}$. $S_1$, Setting 1; $S_2$, Setting 2, $S_3$, Setting 3; M, Method; S, semiparametric bootstrap; P, parametric bootstrap. All numerical entries are multiplied by 1000. }
		\label{tab:SPI_spa}
	\end{table}
In contrast, \cite{hall2006nonparametric} propose a bias reduction with the aid of a double-bootstrap
	$	\mathrm{MSE}^{**}_{B2}(\hat{\theta}^{**}_j)=E^{**}\left( \hat{\theta}^{**}_j-\theta^{**}_j\right) $. In this bootstrapping scheme, for each sample $b$ we must generate $c=1,...,C$ bootstrap samples (in practice, $C=1$ works quite well), where
	\begin{align*}\label{eq:theta_boot2}
		\theta^{**}_j = k^T_j\beta^{**}+l^T_ju^{**}_j, 	
		& &\tilde{\theta}^{**}_j = \theta_j(\delta^{**}) = k^T_j\tilde{\beta}^{**}+l^T_j\tilde{u}^{**}_j,
		& &\hat{\theta}^{**}_j = \theta_j(\hat{\delta}^{**}) = k^T_j\hat{\beta}^{**}+l^T_j\hat{u}^{**}_j.
	\end{align*}
	We can thus consider double bootstrap bias-corrected MSE estimator which is defined as follows 
	\begin{equation*}\label{eq:MSE_bc1}
		\mathrm{MSE}^*_{BC}(\hat{\theta}^*_j) = 2 \mathrm{MSE}^*_{B1}(\hat{\theta}^*_j)- \mathrm{MSE}^{**}_{B2}(\hat{\theta}^{**}_j)  \ .
	\end{equation*}
	Tables \ref{tab:CI_bc}-\ref{tab:SPI_bc} display the performance of individual and simultaneous intervals constructed using $MSE^*_{BC}$. 
	\begin{table}[htb]
		\def~{\hphantom{0}}
		\centering
		{\begin{tabular}{|cccccccccccc|} \hline
				$e_{ij}$      & $u_j$         &   & \multicolumn{3}{c}{Coverage} & \multicolumn{3}{c}{Length} & \multicolumn{3}{c|}{Variance of length} \\
				&               & M & $S_1$    & $S_2$   & $S_3$   & $S_1$   & $S_2$   & $S_3$  & $S_1$    & $S_2$   & $S_3$   \\\hline
				$N(0.5)$      & $N(1)$        & S & 943      & 947     & 947     & 1181    & 852     & 701    & 7        & 1       & 1       \\
				&               & P & 944      & 947     & 946     & 1181    & 852     & 701    & 7        & 1       & 1       \\
				$N(1)$        & $N(0.5)$      & S & 946      & 946     & 947     & 1521    & 1129    & 948    & 12       & 3       & 1       \\
				&               & P & 939      & 946     & 946     & 1488    & 1128    & 947    & 15       & 3       & 1       \\
				$t_6(0.5)$    & $t_6(1)$      & S & 943      & 946     & 948     & 1175    & 853     & 702    & 14       & 2       & 1       \\
				&               & P & 943      & 945     & 947     & 1170    & 849     & 700    & 13       & 2       & 1       \\
				$t_6(1)$      & $t_6(0.5)$    & S & 946      & 945     & 948     & 1501    & 1126    & 947    & 18       & 4       & 2       \\
				&               & P & 937      & 945     & 947     & 1467    & 1123    & 945    & 26       & 4       & 2       \\
				$\chi_5(0.5)$ & $\chi_5((1)$  & S & 941      & 944     & 945     & 1200    & 859     & 705    & 15       & 2       & 1       \\
				&               & P & 945      & 948     & 949     & 1201    & 861     & 706    & 15       & 2       & 1       \\
				$\chi_5(1)$   & $\chi_5(0.5)$ & S & 943      & 944     & 945     & 1590    & 1173    & 974    & 24       & 4       & 2       \\
				&               & P & 943      & 948     & 948     & 1591    & 1176    & 976    & 24       & 4       & 2       \\
				$t_6(0.5)$    & $\chi_5(1)$   & S & 943      & 944     & 945     & 1590    & 1173    & 974    & 24       & 4       & 2       \\
				&               & P & 943      & 948     & 947     & 1591    & 1176    & 706    & 24       & 4       & 1       \\
				$\chi_5(0.5)$ & $t_6(1)$      & S & 940      & 944     & 945     & 1173    & 849     & 699    & 14       & 2       & 1       \\
				&               & P & 946      & 949     & 948     & 1172    & 849     & 701    & 13       & 2       & 1       \\
				$t_6(1)$      & $\chi_5(0.5)$ & S & 943      & 946     & 947     & 1591    & 1179    & 978    & 24       & 5       & 2       \\
				&               & P & 943      & 947     & 817     & 1588    & 1174    & 1542   & 24       & 4       & 256  \\  \hline
		\end{tabular}}
		\caption{Empirical coverage, width and variance of widths of individual intervals at $\alpha = 0.05$-level, $\sigma^2 = MSE^*_{BC}$. 	$S_1$, Setting 1; $S_2$, Setting 2, $S_3$, Setting 3; M, Method; S, semiparametric bootstrap; P, parametric bootstrap. All numerical entries are multiplied by 1000. }
		\label{tab:CI_bc}
	\end{table}
	
	\begin{table}[htb]
		\def~{\hphantom{0}}
		\centering
		{\begin{tabular}{|cccccccccccc|}\hline
				$e_{ij}$      & $u_j$         &   & \multicolumn{3}{c}{Coverage} & \multicolumn{3}{c}{Length} & \multicolumn{3}{c|}{Variance of length} \\
				&               & M & $S_1$    & $S_2$   & $S_3$   & $S_1$   & $S_2$   & $S_3$  & $S_1$    & $S_2$   & $S_3$   \\\hline
				$N(0.5)$      & $N(1)$        & S & 928      & 947     & 938     & 1865    & 1433    & 1221   & 22       & 5       & 3       \\
				&               & P & 926      & 940     & 942     & 1873    & 1439    & 1224   & 22       & 7       & 4       \\
				$N(1)$        & $N(0.5)$      & S & 918      & 942     & 940     & 2392    & 1903    & 1654   & 35       & 12      & 7       \\
				&               & P & 911      & 941     & 948     & 2365    & 1905    & 1654   & 45       & 12      & 8       \\
				$t_6(0.5)$    & $t_6(1)$      & S & 901      & 941     & 946     & 1925    & 1504    & 1278   & 63       & 19      & 10      \\
				&               & P & 876      & 922     & 908     & 1854    & 1434    & 1223   & 38       & 9       & 5       \\
				$t_6(1)$      & $t_6(0.5)$    & S & 920      & 927     & 934     & 2436    & 1979    & 1721   & 71       & 31      & 17      \\
				&               & P & 864      & 907     & 907     & 2332    & 1897    & 1651   & 70       & 17      & 9       \\
				$\chi_5(0.5)$ & $\chi_5((1)$  & S & 907      & 916     & 931     & 2013    & 1557    & 1317   & 69       & 15      & 6       \\
				&               & P & 884      & 859     & 890     & 1903    & 1455    & 1234   & 42       & 10      & 5       \\
				$\chi_5(1)$   & $\chi_5(0.5)$ & S & 918      & 914     & 927     & 2624    & 2099    & 1804   & 98       & 26      & 13      \\
				&               & P & 897      & 878     & 887     & 2525    & 1986    & 1705   & 68       & 18      & 9       \\
				$t_6(0.5)$    & $\chi_5(1)$   & S & 918      & 914     & 927     & 2624    & 2099    & 1804   & 98       & 26      & 13      \\
				&               & P & 897      & 878     & 913     & 2525    & 1986    & 1234   & 68       & 18      & 5       \\
				$\chi_5(0.5)$ & $t_6(1)$      & S & 918      & 917     & 932     & 1966    & 1536    & 1302   & 61       & 14      & 6       \\
				&               & P & 893      & 866     & 902     & 1859    & 1434    & 1223   & 37       & 9       & 5       \\
				$t_6(1)$      & $\chi_5(0.5)$ & S & 915      & 925     & 932     & 2605    & 2078    & 1779   & 102      & 41      & 22      \\
				&               & P & 868      & 900     & 825     & 2521    & 1983    & 1543   & 71       & 18      & 256    \\\hline
		\end{tabular}}
		\caption{Empirical coverage, width and variance of widths of simultaneous intervals at $\alpha = 0.05$-level, $\sigma^2 = MSE^*_{BC}$. $S_1$, Setting 1; $S_2$, Setting 2, $S_3$, Setting 3; M, Method; S, semiparametric bootstrap; P, parametric bootstrap. All numerical entries are multiplied by 1000. }
		\label{tab:SPI_bc}
	\end{table}
	To sum up, the performance of individual and simultaneous intervals is not strongly affected by the choice of the estimator of $\hat{\sigma}^2_j$. The most important factors in the performance of our method is the statistic we are trying to estimate and the appropriate bootstrap method. 
	
\bibliography{citexx3}  

\end{document}